\documentclass[12pt]{article}
\pdfoutput=1 
\usepackage[letterpaper, margin=1in, includehead, includefoot]{geometry}

\newcommand{\blind}{0}
\newcommand{\arxiv}{1}

\usepackage{amsmath,amsfonts,amsthm, algorithm}
\newcommand{\RR}{\mathbb{R}} 
\DeclareMathOperator{\Prob}{Pr} 
\newcommand{\EX}{\mathbb{E}} 
\usepackage{algpseudocode}
\usepackage{dirtytalk} 

\usepackage{natbib}

\usepackage[dvipsnames]{xcolor}
\usepackage{graphicx}
\graphicspath{ {./pictures/} }
\usepackage{caption}
\usepackage{subcaption}
\usepackage{float}

\usepackage{hyperref} 
\usepackage[noabbrev,nameinlink]{cleveref}
\newtheorem{theorem}{Theorem}
\newtheorem{proposition}[theorem]{Proposition}
\theoremstyle{remark}
\newtheorem*{remark*}{Remark} 

\usepackage[title,titletoc,toc]{appendix}

\begin{document}
\def\spacingset#1{\renewcommand{\baselinestretch}%
{#1}\small\normalsize} \spacingset{1}

\if0\blind
{
  \title{\bf Cohesion and Repulsion in Bayesian Distance Clustering}
  \renewcommand\footnotemark{}
  \author{Abhinav Natarajan\thanks{
  Abhinav Natarajan is a doctoral student at the Mathematical Institute, University of Oxford, United Kingdom OX1 2JD (E-mail: \emph{\href{mailto:natarajan@maths.ox.ac.uk}{natarajan@maths.ox.ac.uk}}). Maria De Iorio is Professor at Yong Loo Lin School of Medicine, National University of Singapore, Singapore 138527, Professor of Biostatistics at University College London and Principal Investigator at A*STAR, Singapore 138632 (E-mail: \emph{\href{mailto:mdi@nus.edu.sg}{mdi@nus.edu.sg}}). Andreas Heinecke is Assistant Professor of Mathematics at Yale-NUS College, Singapore 138527 (E-mail: \emph{\href{mailto:andreas.heinecke@yale-nus.edu.sg} {andreas.heinecke@yale-nus.edu.sg}}). Emanuel Mayer is Associate Professor of History at Yale-NUS College, Singapore 138527 (E-mail: \emph{\href{mailto:emanuel.mayer@yale-nus.edu.sg}{emanuel.mayer@yale-nus.edu.sg}}). Simon Glenn is Research Fellow at the Ashmolean Museum, University of Oxford, United Kingdom OX1 2PH (E-mail: \emph{\href{mailto:simon.glenn@ashmus.ox.ac.uk}{simon.glenn@ashmus.ox.ac.uk}}).}\\
  University of Oxford
  \and 
  Maria De Iorio \\ National University of Singapore
  \and
  Andreas Heinecke \\ Yale-NUS College
  \and
  Emanuel Mayer \\ Yale-NUS College
  \and 
  Simon Glenn \\ University of Oxford
  }
  \date{March 2023}
  \maketitle
} \fi

\if1\blind
{
    \bigskip
    \bigskip
    \bigskip
    \begin{center}
        {\LARGE\bf Cohesion and Repulsion in Bayesian Distance Clustering}
    \end{center}
    \medskip
} \fi

\begin{abstract}
Clustering in high-dimensions poses many statistical challenges. While traditional distance-based clustering methods are computationally feasible, they lack probabilistic interpretation and rely on heuristics for estimation of the number of clusters. On the other hand, probabilistic model-based clustering techniques often fail to scale and devising algorithms that are able to effectively explore the posterior space is an open problem. Based on recent developments in Bayesian distance-based clustering, we propose a hybrid solution that entails defining a likelihood on pairwise distances between observations. The novelty of the approach consists in including both cohesion and repulsion terms in the likelihood, which allows for cluster identifiability. This implies that clusters are composed of objects which have small \emph{dissimilarities} among themselves (cohesion) and similar dissimilarities to observations in other clusters (repulsion). We show how this modelling strategy has interesting connection with existing proposals in the literature. The proposed method is computationally efficient and applicable to a wide variety of scenarios. We demonstrate the approach in simulation and an application in digital numismatics. Supplementary Material with code is available online. 
\end{abstract}

\noindent%
\vfill

\if0\arxiv
\spacingset{1.5} 
\fi

\section{Introduction}
Multidimensional clustering has been a fruitful line of research in statistics for a long time. The surge in the availability of data in recent years poses new challenges to clustering methods and the scalability of the associated computational algorithms, particularly in high dimensions. There are two main classes of clustering methods: those based on probabilistic models (model-based clustering), and constructive approaches based on \emph{dissimilarities} between observations (distance-based clustering). The first class of methods includes popular tools such as mixture models \citep{McLachlan1988MixtureM,Dasgupta1998MixtureM}, product partition models (PPMs) \citep{Hartigan1990PPM,BarryHartigan1992PPM}, and nonparametric models like the Dirichlet process or more general species sampling models \citep{Pitman1996SSM,Ishwaran2003SSM}. An overview can be found in the article by \citet{Quintana2006OverviewMBC}. The second class of methods includes the popular hierarchical clustering, $k$-means, and its variants like $k$-medoids.

Distance-based clustering algorithms, although computationally accessible and scalable to high dimensions, are often less interpretable, and do not quantify clustering uncertainty because of the lack of a probabilistic foundation. Out-of-sample prediction is challenging with these algorithms, and inference on the number of clusters relies on heuristics such as the elbow method. Moreover, there are theoretical limitations to the results produced by any distance-based clustering algorithm; in particular, they cannot simultaneously satisfy constraints about scale-invariance and consistency while also exploring all possible partitions \citep{Kleinberg2002Impossibility}. On the other hand the drawbacks of model-based clustering methods are their analytic intractability and computational burden arising when working with high dimensional observations. To add to this, a fundamental difficulty with both types of clustering methods is that there is no consensus on what constitutes a true cluster \citep{Hennig2015TrueClusters}, and that the aims of clustering should be application-specific. 

The focus of this paper is high dimensional clustering, in particular when point-wise evaluation of the likelihood is computationally intractable and posterior inference is infeasible. Our approach builds on recent proposals by \citet{Duan2019DistanceClustering} and \citet{Rigon2020GBPPM} that bridge the gap between model-based and distance-based clustering. The main idea behind this research is to specify a probability model on the distances between observations instead of the observations themselves, reducing a multidimensional problem to a low-dimensional one. An early reference for Bayesian clustering based on distances can be found in \cite{Lau2007Clustering}.

Let $\rho_n = \{C_1, \ldots, C_K\}$ denote a partition of the set $[n] = \{1, \ldots, n\}$, and let $\boldsymbol X =\{x_1, \ldots, x_n\}$ be a set of observations in $\RR^l$. It is convenient to represent a clustering through cluster allocation indicators $z_i$, where $z_i = j$ when $i \in C_j$. \citet{Rigon2020GBPPM} reformulate the clustering problem in terms of decision theory. They show that a large class of distance-based clustering methods based on loss-functions, including $k$-means and $k$-medoids, are equivalent to maximum a posteriori estimates in a probabilistic model with appropriately defined likelihood on the distances. Explicitly, they consider product partition models where the likelihood decomposes into cluster-wise \emph{cohesions}
\begin{align*}
\pi(\mathbf{X} \mid \lambda, \rho_n) = \prod_{k=1}^K \exp\left(-\lambda \sum_{i \in C_k} D(x_i, C_k)\right)
\end{align*}
where $D(x_i, C_k)$ measures the dissimilarity of observation $i$ from cluster $C_k$ and $\lambda$ is a parameter that controls the posterior dependence on the distances between observations. A major drawback in this approach is that the number of clusters $K$ must be pre-specified, whereas $K$ is an object of inference in many practical scenarios. Inference on $K$ is also problematic in the method proposed by \citet{Duan2019DistanceClustering}. This is due to identifiability issues that arise when working with distances. The starting point of their approach is an overfitted mixture model. By noting that in high dimensions the contribution of the cluster centres to the likelihood is negligible compared to the contribution from pairwise distances within the cluster, they specify a \emph{partial likelihood} on the pairwise distances between observations
\begin{align*}
\pi(\mathbf{X} \mid \rho_n, \alpha, \beta) = \prod_{k=1}^K \prod_{i, j \in C_k} g(d(x_i, x_j); \alpha, \beta)^{1/n_k}
\end{align*}
where $g$ is a $\mathrm{Gamma}(\alpha, \beta)$ density and $n_k$ is the size of the $k$th cluster. Although this approach allows for estimation of $K$, it often relies on the specification of the maximum number of clusters in the sense that the clustering allocation significantly changes with this parameter. 

We propose a model for high-dimensional clustering based on pairwise distances that combines cluster-wise cohesions with a \emph{repulsive} term that imposes a strong identifiability constraint in the likelihood by penalising clusters that are not well-separated. To this end, we borrow ideas from machine learning such as the cross-cluster penalty in the calculation of a silhouette coefficient, and from the literature on repulsive distributions. The idea of repulsive distributions has been previously studied in the context of mixture models \citep{Petralia2012Repulsive,Quinlan2017Repulsive,Xu2016RepulsionDPP} to separate the location and scale parameters of the mixture kernels. We discuss the connection of repulsive distributions to our model in more detail in \Cref{sec: likelihood specification}. 

There are other instances of model-based clustering methods which exploit pairwise distances for cluster estimation. One example is the framework of Voronoi tesselations, a partition strategy that has found application in Bayesian statistics and partition models \citep{Denison2001BayesianPartitioning,Moller2001Voronoi,Corander2008Voronoi}. In this approach a set of centres is sampled from a prior and the sample space is partitioned into the associated Voronoi cells. When the centres are chosen from the observations themselves, the implied prior on partitions depends on the pairwise distances between the observations. In the Bayesian random partition model literature, there have been various proposals to include covariate information in cluster allocation probabilities. Most notably, \citet{Muller2011PPMX} use a \emph{similarity function} defined on sets of covariates belonging to all experimental units from a given cluster to modify the cohesions of a product partition model. Their similarity function is the marginal density of the covariates from an auxiliary probability model, which can also be interpreted as the marginal density on the distances of the covariates from a latent centre in the auxilliary probability space. This approach incorporates information about the distances between covariates into the partition prior. In high-dimensional settings, i.e., when the number of covariates is large, the covariate information dominates the clustering and the influence of the response is relatively inconsequential. See for example the work by \cite{Barcella2017covariates}. Alternatively, \citet{Dahl2008DistanceBasedPD,Dahl2017RPM} propose random partition models through different modifications of the Dirichlet process cluster allocation probability: in the first case of the full conditional $\Prob(z_i \mid \boldsymbol{z}_{-i})$, and in the second case of the sequential conditional probabilities $\Prob(z_i \mid z_1, \ldots, z_{i-1})$. 

All these methods are linked through the use of pairwise dissimilarities, often in the form of distances, to define a partition prior for flexible Bayesian modelling. Here, we use the same strategy to define the likelihood on pairwise distances while using standard partition priors such as the Dirichlet process or the recently proposed \emph{microclustering priors} \citep{Zanella2016Microclustering,Betancourt2020Microclustering}. In this respect, our model is strongly related to composite likelihood methods which will be discussed in \Cref{sec: likelihood specification}.

The paper is structured as follows. In \Cref{sec: Model} we introduce the model and the computational strategy. In \Cref{sec: coins example} we apply the proposed methodology to a problem from digital numismatics. We conclude the paper in \Cref{sec: discussion}. In Supplementary Material we present details of the computational algorithm, extensive simulation studies, and further results from the data application. 

\section{Model}\label{sec: Model}
In this section we describe the pairwise distance-based likelihood, and we present a justification for our modelling approach. The proposed strategy can accommodate different partition priors. In particular, we discuss a microclustering prior \citep{Betancourt2020Microclustering} as it is the most relevant for our application in digital numismatics.  We conclude the section with a discussion on the choice of hyperparameters, and of the MCMC algorithm. 

\subsection{Likelihood specification}\label{sec: likelihood specification}
We specify the likelihood on pairwise distances between observations instead of directly on the observations. This strategy falls naturally into the framework of composite likelihood. In its most general form, a composite likelihood is obtained by multiplying together a collection of component functions, each of which is a valid conditional or marginal density \citep{Lindsay1988Composite}. The utility of composite likelihoods is in their computational tractability when a full likelihood is difficult to specify or computationally challenging to work with. In this context, the working assumption is the conditional independence of the individual likelihood components. Key examples of composite likelihood approaches include pseudolikelihood methods for approximate inference in spatial processes \citep{Besag1975Pseudolikelihood},
posterior inference in population genetics models \citep{Li2003CompositeGenetics,Larribe2011composite},
pairwise difference likelihood and maximum composite likelihood in the analysis of dependence structure \citep{LeleTaper2002Composite}, and the use of independence loglikelihood for inference on clustered data \citep{Chandler2007IndependenceLogLikelihood}. See \citet{Varin2011Composite} for an overview. Other approaches to overcome likelihood intractability include specifying the likelihood on summary statistics of the data \citep{Beaumont2002ABC}, or comparing simulated data from the model with the observed data \citep{Fearnhead2012ABC}. Both these ideas underlie \emph{approximate Bayesian computation} \citep{Marjoram2003ABC}. 

Combining ideas from composite likelihood methods and distance-based clustering, our strategy is to specify a likelihood on the distances that decomposes into a contribution from within-cluster distances and cross-cluster distances:
\begin{equation}\label{eqn: likelihood general form}
\pi(\boldsymbol{D} \mid \boldsymbol{\theta}, \boldsymbol{\lambda}, \rho_n) = 
\left[\prod_{k=1}^{K} \prod_{\substack{i, j \in C_k \\ i <  j}} f(D_{ij} \mid \lambda_k) \right] \left[ \prod_{(k, t) \in A} \prod_{\substack{i \in C_{k}\\ j \in C_{t}}} g(D_{ij} \mid \theta_{kt}) \right]
\end{equation}
where $\boldsymbol{D} = [d(x_i, x_j)]_{ij}$ is the matrix of all pairwise distances, $A = \{(k, t) : 1 \leq k < t \leq K\}$, and $f$ and $g$ are probability densities. Note that this formulation does not result in a valid probability model on the data, but rather on a space $\mathcal{X}$ that is obtained as follows: let $G$ be the group of isometries of $\RR^l$ (with respect to the chosen distance metric), and let $H = \{(g, \ldots, g) \in G^n: g \in G\}$ be the diagonal subgroup of $G^n$. Then $\mathcal{X}$ is the orbit space $\RR^{l \times n} / H$ (for a reference see \citealt{Janich1995Topology}). In \Cref{sec: choice of likelihood densities} we discuss the choice of $f$ and $g$ in \Cref{eqn: likelihood general form}. The first term in \Cref{eqn: likelihood general form} is similar to the cohesions of \citet{Duan2019DistanceClustering,Rigon2020GBPPM} and quantifies how similar the observations within each cluster are to each other; we call this the \emph{cohesive} part of the likelihood. 

The second multiplicative term in the likelihood, which we call the \emph{repulsive} term, is related to the idea of repulsive mixtures. Typical mixture models associate with each cluster a location parameter $\phi_j$, and these are assumed to be i.i.d. from a fixed prior. \citet{Petralia2012Repulsive}, and \citet{Quinlan2017Repulsive} relax the i.i.d. assumption and use a repulsive joint prior of the form
\begin{equation}\label{eqn: Petralia repulsive}
\pi(\boldsymbol{\phi}) \propto \prod_{i, j} h(d(\phi_{i}, \phi_j))
\end{equation}
where $d$ is a distance measure and $h$ decays to 0 for small values of its input. They do this to penalise clusters that are too close to each other, inducing parsimony. We generalise this idea by using a repulsive distribution on the observations themselves, i.e., by setting $g$ in \Cref{eqn: likelihood general form} to a density that decays as its input approaches 0. This form of repulsion is important to our application because it encourages the formation of clusters from points that are not only close to each other but also have similar distances to points in other clusters. Moreover the repulsion allows for inter-cluster distances of different magnitude for different pairs of clusters. Consequently, this strategy allows for estimation of the number of clusters. Using repulsion on the observations instead of the cluster centres is also a viable strategy when the location parameters are not of interest or when posterior inference on the location parameters is computationally difficult, as is usually the case in high-dimensional clustering \citep{Johnstone2009Clustering}. In doing so, we relax the assumption of conditional independence between clusters given their cluster-specific parameters. In Supplementary Material we further investigate the role of the repulsion term, showing its importance in identifying the number of clusters. This is consistent with the work by \cite{Fuquene2019Nonlocal} that shows that repulsion leads to faster learning of $K$ in model-based settings.

\subsection{Posterior Uncertainty}\label{sec: posterior uncertainty}
The distance based likelihood in \Cref{eqn: likelihood general form} has sharper peaks and flatter tails than the model-based likelihood from the raw data, as is typical in composite likelihood or pseudo-likelihood frameworks due to the artificial independence assumption. Intuitively, our model assumes $O(n^2)$ independent pieces of information whereas the data generation process may only produce $O(n)$ independent pieces of information. Consequently, well-separated clusters are associated to high posterior probability with corresponding underestimation of uncertainty, while poorly separated clusters will often be split into smaller clusters due to the artificially increased uncertainty. Although the estimation of uncertainty is inaccurate, localisation of modes in the posterior is satisfactory (as is typically the case for composite likelihood methods) and the model is able to provide some measure of uncertainty even when direct approaches fail to recover the clustering structure. This is demonstrated in our simulations in Supplementary Material. Moreover, our model can be used to guide more direct model-based approaches with better prior information. We believe that the major drawback of our approach is its dependence on a choice of dissimilarity measure, as we remark in the discussion section. Finally we note that a common strategy in composite likelihood models to counteract the underestimation of uncertainty is to artificially flatten the likelihood, raising it to the power $1/n$. We cannot employ the same approach as this would flatten both the within-cluster terms and the cross-cluster terms, making them overlap significantly and making the clusters unidentifiable. We have nevertheless tried this approach, and as expected we obtained poor inference results (not shown).

\subsection{Choice of distance densities}\label{sec: choice of likelihood densities}
The likelihood in \Cref{eqn: likelihood general form} results in a monotonically decreasing density on the within-cluster distances if the cohesive term is chosen as the exponential of a loss function, as suggested by \citet{Rigon2020GBPPM}. This choice might be too restrictive in application, as more flexible distributions are required to accommodate the complexity in the data. As a consequence of such a restrictive choice, more dispersed clusters may be broken up into smaller clusters. \citet{Rigon2020GBPPM} alleviate this problem by fixing the number of clusters. We instead propose a more flexible choice of $f$ and $g$ motivated by the following commonly-encountered scenario.

Assume that the original data have a multivariate Normal distribution such that each cluster is defined by a Normal kernel
\begin{align*}
y_i \mid z_i = k, \mu_k, \sigma_k^2 \sim \mathcal{N}(\mu_k, \sigma_k^2I_l)
\end{align*}
Hence the within-cluster differences are distributed as $\mathcal{N}(0, 2\sigma_k^2 I_l)$ and the corresponding squared Euclidean distances have a $\mathrm{Gamma}(l/2, 1/(2\sigma_k^2))$ distribution. On the other hand, inter-cluster squared distances are distributed as a 3-parameter non-central $\chi^2$:
\begin{multline*}
g\left(\| x_i - x_j\|_2^2 \mathrel{\big|} z_i = k, z_j = t, k\neq t, \theta_{kt} =  \left(\theta_{kt}^{(1)}, \theta_{kt}^{(2)}\right)\right) =\\ \sum_{m=0}^\infty \frac{\left(\theta_{kt}^{(1)}\right)^m\exp\left(-\theta_{kt}^{(1)}\right)}{m!} h\left(\| x_i - x_j\|_2^2; l/2 + m, \theta_{kt}^{(2)}\right)
\end{multline*}
where $\theta^{(1)}_{kt}$ is the noncentrality parameter and corresponds to the squared distance between the cluster centres $\mu_k$ and $\mu_t$, $\theta^{(2)}_{kt}$ is a scale parameter related to the within-cluster variances of the two clusters, $l$ is the dimension of the original data, and $h(\cdot; a, b)$ is a $\mathrm{Gamma}(a, b)$ density. This setup would cover many real world applications, but posterior inference on the parameters of a non-central $\chi^2$ is unnecessarily complicated. Moreover the non-central $\chi^2$ is defined on the squared Euclidean distance, which will lead to a non-central $\chi$ distribution on the distances. Indeed when we know that the data generation process coincides with a Normal mixture, we should use the correct distribution but this is not often the case. 
As such when the data generating process is unknown we work with Gamma distributions on the distances mainly for computational convenience. We propose setting $f$ to be a $\mathrm{Gamma}(\delta_1, \lambda_k)$ as in \citet{Duan2019DistanceClustering}:
\begin{align*}
f(x \mid \lambda_k) = \frac{\lambda_k^{\delta_1} x^{\delta_1-1}\exp(-\lambda_k x)}{\Gamma(\delta_1)}
\end{align*}
where $x$ is a pairwise distance and $\delta_1$ is a fixed shape parameter that controls the cluster dispersion. When $\delta_1 < 1$, $f$ is a monotonically decreasing density. We set $g$ in \Cref{eqn: likelihood general form} to be a $\mathrm{Gamma}(\delta_2, \theta_{kt})$ density, where $\delta_2 > 1$ is a fixed shape parameter that controls the shape of the decay of $g$ towards the origin. We note that the constraint $\delta_2 > 1$ ensures that $g$ is a repulsive density by forcing $g(0 \mid \theta_{kt}) = 0$. An appropriate choice of $\delta_1$ and $\delta_2$ is application-specific, and we discuss possible alternatives in \Cref{sec: prior specification}.

In our experiments we find that using a Gamma distribution directly on the distances does not have an appreciable effect on posterior inference, suggesting that the methodology is robust. This strategy is also followed by \citet{Duan2019DistanceClustering} but for different reasons. 

\subsection{Prior specification}\label{sec: prior specification}
Here we discuss the choice of priors for the cluster-specific parameters and the partition.

\subsubsection{Prior Cluster-Specific Parameters} \label{sec: cluster-specific parameters}
For computational convenience, we choose conjugate Gamma priors $\lambda_k \stackrel{\mathrm{iid}}{\sim} \mathrm{Gamma}(\alpha, \beta)$ and $\theta_{kt} \stackrel{\mathrm{iid}}{\sim} \mathrm{Gamma}(\zeta, \gamma)$. To set $\delta_1, \delta_2, \alpha, \beta, \zeta$ and $\gamma$, we follow a procedure in the spirit of empirical Bayes methods, as straightforward application of empirical Bayes is hindered by an often flat marginal likelihood of the parameters in question. We summarise our method in \Cref{alg: prior on cluster-specific parameters}.
\begin{algorithm}[h]
\spacingset{1.25}
\caption{Choosing $\alpha, \beta, \zeta, \gamma, \delta_1$, and $\delta_2$}
\label{alg: prior on cluster-specific parameters}
\begin{enumerate}
\item Compute a heuristic initial value for $K$, say $K_{\text{elbow}}$, via the elbow method, fitting $k$-means clustering for a range of values of $K$ and with the within-cluster-sum-of-squares (WSS) score as the objective function. 
\item Use $k$-means clustering with $K_{\text{elbow}}$ to obtain an initial clustering configuration. 
\item Split the pairwise distances into two groups $A$ and $B$ that correspond to the within-cluster and inter-cluster distances in this initial configuration.
\item Fit a Gamma distribution to the values in $A$ using maximum likelihood estimation and set $\delta_1$ to be the shape parameter of this distribution.
\item Set $\alpha = \delta_1 n_A$ and $\beta = \sum_{a \in A} a$, where $n_A$ is the cardinality of the set $A$. This corresponds to the conditional posterior of $\lambda$ obtained by specifying an improper prior $\pi(\lambda) \propto I(\lambda > 0)$ and treating $A$ as a weighted set of observations from a $\mathrm{Gamma}(\delta_1, \lambda)$ distribution.
\item Repeat steps 4 and 5 to obtain values for $\delta_2$, $\zeta$ and $\gamma$ by considering the values in $B$.
\end{enumerate}
\end{algorithm}

The range for $K$ in step 1 of the algorithm can be chosen to be quite broad, for example from one to $n-1$. When only pairwise distances or dissimilarities are available and not the raw data, $k$-medoids and the within-cluster-sum-of-dissimilarities can be used instead. 

The proposed method depends on the choice of $K$ obtained by the elbow method. In Supplementary Material we show that posterior inference is robust to the choice of $K$ obtained, as long as this choice lies within a sensible range. We also propose an alternative method to fit the prior that results in a mixture prior on possible values for $K$.

\subsubsection{Prior on partitions}\label{sec: prior on partitions}
The model can accommodate any prior on partitions of the observations, which is equivalent to specifying a prior on the partitions of $[n] = \{1, \ldots, n\}$. Let $\rho_n = \{C_1, \ldots, C_K\}$ denote a partition of $[n]$ where the $C_j$ are pairwise disjoint and $K \leq n$. A common choice is to use a product partition model (PPM) as the prior for $\rho_n$; see for example the paper by \citet{Hartigan1990PPM} or \citet{BarryHartigan1992PPM}. In a PPM there is a non-negative function $c(C_j)$, usually referred to as a \emph{cohesion function}, which is used to define the prior 
\begin{align*}
\Prob(\rho_n) = M \prod_{j=1}^K c(C_j)
\end{align*}
where $M$ is a normalising constant. This prior includes as special cases the Dirichlet Process \citep{Quintana2003PPM} as well as Gibbs-type priors. Alternatively one can consider the implied prior on partitions derived from a species sampling model \citep{Pitman1996SSM}; in this case it can be shown that $\Pr(\rho_n = \{C_1, \ldots, C_K\}) = p(n_1, \ldots, n_K)$ where $n_j = |C_j|$ is the number of elements in $C_j$ and $p$ is a symmetric function of its arguments called the \emph{exchangeable partition probability function} (EPPF). 

In our application, we opt for a prior that has the microclustering property \citep{Miller2015Microclustering,Zanella2016Microclustering,Betancourt2020Microclustering}; that is, cluster sizes grow sublinearly in the number of observations $n$. This property is appropriate for die analysis in numismatics where each die is represented by a very limited number of samples. We use a class of random partition models described in \citet{Betancourt2020Microclustering} called \emph{Exchangeable Sequence of Clusters (ESC)}. In this model a generative process gives rise to a prior on partitions, which we describe briefly. A random distribution $\nu$ is drawn from the set $\mathcal{P}$ of distributions on the positive integers; $\nu$ is distributed according to some $P_\nu$. The cluster sizes $n_j$ are sampled from $\nu$, conditional upon the following event
\begin{align*}
E_n = \left\{\text{there exists $K \in \mathbb{N}$ such that } \sum_{j=1}^K n_j = n\right\}.
\end{align*}
We require that $\nu(1) > 0$ for all $\nu$ in the support of $P_{\nu}$ to ensure that $\Prob(E_n \mid \nu) > 0$ for all $\nu$. A random partition with cluster sizes $\{n_1, \ldots, n_K\}$ is drawn by allocating cluster labels from a uniform permutation of 
\begin{align*}
(\underbrace{1, \ldots, 1}_{n_1 \text{ times}}, \underbrace{2, \ldots, 2}_{n_2 \text{ times}}, \ldots, \underbrace{K, \ldots, K}_{n_K \text{ times}})
\end{align*}
The resulting partition model is denoted $ESC_n(P_\nu)$. Here we give details on the clustering structure implied by the microclustering prior in \Cref{eqn: Betancourt proposition 2a,eqn: Betancourt proposition 2b}, as well as on the prior predictive distribution of the cluster label for a new observation in \Cref{eqn: Betancourt corollary 1}. \citet{Betancourt2020Microclustering} derive the conditional and marginal EPPF for the class of microclustering priors as well as the conditional allocation probabilities. Let $(z_1, \ldots, z_n)$ be the cluster allocation labels for $\rho_n \sim ESC_n(P_\nu)$. Then for any $i \in [n]$ \citet{Betancourt2020Microclustering} show that:
\begin{align}
\Prob(\rho_n \mid \nu) = \Prob(n_1, \ldots, n_K \mid \nu) &=
\frac{K!}{n!\Prob(E_n \mid \nu)}\prod_{j=1}^K n_j! \nu(n_j) \label{eqn: Betancourt proposition 2a}
\\
\Prob(\rho_n) = \Prob(n_1, \ldots, n_K) &= \displaystyle
\frac{1}{\Prob(E_n)} \EX_{\nu \sim P_\nu} \left[\frac{K!}{n!}\prod_{j=1}^K n_j! \nu(n_j)\right] \label{eqn: Betancourt proposition 2b}\\
\pi(z_i = j \mid \mathbf{z}_{-i}, \nu) &\propto \begin{cases}
(n_{j, -i}+1)\displaystyle \frac{\nu(n_{j, -i} + 1)}{\nu(n_{j, -i})} & j = 1, \ldots, K_{-i}\\
(K_{-i} + 1) \nu(1) & j = K_{-i} + 1
\end{cases}\label{eqn: Betancourt corollary 1}
\end{align}
where $\boldsymbol{z}_{-i}$ is the set of cluster labels excluding $z_i$, $n_{j, -i}$ is the numerosity of $C_{j, -i} = C_j \setminus \{i\}$, and $K_{-i}$ is the number of clusters in the induced partition of $[n] \setminus \{i\}$. \citet{Betancourt2020Microclustering} suggest setting $\nu$ to a negative binomial truncated to the positive integers and show that the resulting model, which they call the ESC-NB model, exhibits the microclustering property. We use a variant of the ESC-NB model by setting $\nu$ to a shifted negative binomial as it aids the choice of hyperparameters. We set $\nu = \mathrm{NegBin}(r, p) + 1$ where $r \sim \mathrm{Gamma}(\eta, \sigma)$ and $p \sim \mathrm{Beta}(u, v)$. To set the hyperparameters $\sigma, \eta, u$ and $v$, one can use the conditional distribution on the number of clusters $K$ and a prior guess on the number of clusters. In general posterior inference is not sensitive to the choice of the hyperparameters in the prior for $r$ and $p$. Nevertheless, the marginal and conditional distributions of $K$ can be analytically calculated or approximated as in the following proposition. The proposition can be used for setting the hyperparameters in the priors for $r$ and $p$, especially when relevant prior information on $K$ is available.
\begin{proposition}\label{thm: distributions on K}
The conditional distribution on the number of clusters $K$ in the ESC model with a shifted negative binomial is given by 
\begin{align}
\pi(K \mid r, p) &= \frac{1}{\Prob(E_n \mid r, p)} \begin{cases}
\displaystyle \frac{(1-p)^{rK}p^{n-K}}{(n-K)B(rK, n-K)} & K < n\\
(1-p)^{rn} & K = n
\end{cases}\label{eqn: conditional distribution on K in ESC-NB model}
\end{align}
where $B(\cdot, \cdot)$ is the Beta function. The marginal distribution of $K$ is approximated by 
\begin{equation}\label{eqn: marginal distribution on K in ESC-NB model approximation}
\pi(K) \approx \tilde \pi(K) \propto \frac{\Gamma(n-K+u)}{\Gamma(n-K+1)}\times \begin{cases} 
\sigma^\eta \Gamma(\eta + v) K^{-\eta}(n-K)^{\eta-u} \Psi(v+\eta, \eta - u + 1, \sigma/\omega_K) & K < n\\
\sigma^{u}\Gamma(\eta-u) K^{-u} & K = n
\end{cases}
\end{equation}
where $\Psi(\cdot, \cdot, \cdot)$ is the confluent hypergeometric function of the second kind. If $u = v = 1$, the marginal distribution of $K$ is exactly given by
\begin{equation}\label{eqn: marginal distribution on K in ESC-NB model special case}
\pi(K) \propto \omega_K'^{-\eta} \Psi\left(\eta, \eta, \frac{\sigma}{\omega_K'}\right) - \mathbf{1}_{\{K < n\}}\omega_K^{-\eta} \Psi\left(\eta, \eta, \frac{\sigma}{\omega_K}\right)
\end{equation}
where $\omega_K = \frac{K}{n-K}$ and $\omega_K' = \frac{K}{n-K+1}$.
\end{proposition}
\begin{proof}
See Supplementary Material. 
\end{proof}

In our simulation studies and real data analyses, we opt for an empirical Bayes approach (see \Cref{alg: prior on partitions}) to set the hyperparameters for $r$ and $p$, which is consistent with our method for setting the hyperparameters for $\boldsymbol\lambda$ and $\boldsymbol \theta$. 
\begin{algorithm}[h]
\spacingset{1.25}
\caption{Choosing values for $\eta, \sigma, u$, and $v$.}
\label{alg: prior on partitions}
\begin{enumerate}
\item Fix the cluster labels at the initial clustering configuration obtained in Step 2 of \Cref{alg: prior on cluster-specific parameters}. 
\item Sample $r$ and $p$ from their conditional posteriors in the model using a $\mathrm{Gamma}(1, 1)$ prior for $r$ and $\mathrm{Beta}(1, 1)$ prior for $p$. 
\item Use MLE to fit a $\mathrm{Gamma}(\eta, \sigma)$ distribution to the posterior samples of $r$ and a $\mathrm{Beta}(u, v)$ distribution to the posterior samples of $p$.
\end{enumerate}
\end{algorithm}

We conclude this section by noting that the model lends itself to any choice of partition prior. Particular choices could favour a different clustering structure and should be tailored to the application in question. For instance, Pitman-Yor \citep{PitmanYor1997PitmanYorPrior} or Gibbs-type priors \citep{Gnedin2006ExchangeableGP} could be used as drop-in replacements for the microclustering prior and cover a wide range of partition priors such as mixture with random number of components \citep{Miller2018Mixture,Argiento2022Infinity}. When prior information on the partition is available a more direct approach could be employed; see for example \citet{Paganin2021InformativePrior}. We also note that as the number of observations $n$ grows larger than the number of clusters $K$ (i.e., when not all clusters are singletons), posterior inference quickly becomes less sensitive to the choice of partition prior as the likelihood will dominate the posterior. 

\subsection{Posterior inference}\label{sec: posterior inference}
Posterior inference is performed through an MCMC scheme. Cluster allocations can be updated either through a Gibbs update of individual cluster labels, or through a split-merge algorithm as proposed by \citet{Neal2004SplitMerge}. The split-merge algorithm is more efficient for large $n$ as it leads to better mixing of the chain. In our applications we combine a Gibbs step and a split-merge step in each iteration as suggested by \cite{Neal2004SplitMerge}. In Supplementary Material we show that the time complexity of the cluster reallocation step is $O(n^2)$. We note that the pre-computation of the $\displaystyle\binom{n}{2}$ pairwise distances is typically not a bottleneck, as the distances are computed only once. Posterior inference for $r$ and $p$ are performed through a Metropolis step and a Gibbs update respectively. We do not sample the $\lambda_k$ and $\theta_{kt}$ as we marginalise over them. If required, they can be sampled by conditioning on the cluster allocation and sampling from a Gamma-Gamma conjugate model. See Supplementary Material for derivations of the posterior conditionals and the full details of the Gibbs and split-merge algorithms.

We note that the MCMC scheme can be modified as necessary to accommodate different choice of partition priors, as efficient algorithms are available for most Bayesian nonparametric processes.

\section{Application to Digital Numismatics}\label{sec: coins example}

\subsection{Description of the data}
Die studies determine the number of dies used to mint a discreet issue of coinage. With almost no exceptions, dies were destroyed after they wore out, which is why die studies rely on an analysis of the coins struck by them. From a statistical viewpoint the first task to be accomplished is clustering the coins with the goal of identifying if they were cast from the same die.

Die studies are an indispensable tool for pre-modern historical chronology and economic and political history. They are used for putting coins (and by extension rulers and events) into chronological order, to identify mints, and for estimating the output of a mint over time. 
While digital technology has made large amounts of coinage accessible, numismatic research still requires meticulous and time-consuming manual work. Conducting a die study is time consuming because each coin has to be compared visually to every other coin at least once to determine whether their obverse (front face) and their reverse (back face) were struck from the same dies. For example, a study of 800 coin obverses would require more than 300 000 visual comparisons and could take an expert numismatist approximately 450 hours from scratch. This makes it practically impossible to conduct large scale die studies of coinages like that of the Roman Empire, which would be historically more valuable than the small-scale die studies done today. The practical difficulties of manual die-studies calls for computer-assisted die studies.

We consider here silver coins from one of several issues minted between late 64 C.E and mid 66 C.E., immediately after the great fire of Rome. Pressed for funds, Nero reduced the weight of gold and silver coins by c. 12\%, so that he could produce more coinage out of the available bullion stock. Determining the number of dies used to strike this coinage will make it possible to come up with reasonable estimates of how many gold and silver coins Nero minted during this period, and help to determine how much bullion he may have saved in the immediate aftermath of the great fire. This type of numismatic work would require time-consuming effort by highly trained experts if performed manually. Here we demonstrate the potential in digital numismatics of our strategy by clustering a dataset of 81 coins, which requires a few hours for pre-processing of the images and a few minutes to fit our model. The distance computation is straightforward to parallelise, further speeding up computation. The data consists of 81 high-resolution images of obverses taken from a forthcoming die study on Nero’s coinage. To test the performance of our model, die analysis is first performed by visual inspection by a numismatic expert to provide the ground truth. This analysis identifies ten distinct die groups. The images were standardised to 380 $\times$ 380 pixels to compute the pairwise distances. 

\subsection{Computing pairwise distances}\label{sec: computing pairwise distances for coins}
Fitting the model in \Cref{eqn: likelihood general form} requires the definition of a distance between images that has the potential to differentiate between images of coins minted from different dies and to capture the similarity of images of coins minted from the same die. The pixelwise Euclidean distance between the digital images cannot be used to obtain such information about the semantic dissimilarity of images. Due to the high dimensionality of the ambient image space the data set of images is sparse, with little separation between the largest and smallest pairwise Euclidean distances in the data set \citep{beyer1999nearest}. \Cref{fig: histogram of coin distances pixelwise} illustrates this for our dataset. In contrast, numismatists rely on domain knowledge and often years of experience to identify few key feature points in images of coins to aid comparisons. This essentially coincides with disregarding irrelevant features and performing dimension reduction. When defining the distance between images, our goal is to automate expert knowledge acquisition and focus on extraction of key features. This is a common strategy in many tasks in computer vision \citep{szeliski2010computer} and, more generally, in statistical shape analysis \citep{dryden1998statistical,Gao2019GaussianII}. \citet{taylor2020computer} uses landmarking to define a distance between images of ancient coins with the ultimate goal of die-analysis using simple hierarchical clustering. They do not provide an estimate of the number of dies represented in the sample or an overall subdivision of coins into die groups.
\begin{figure}[H]
\centering
\begin{subfigure}[t]{0.32\textwidth}
    \centering
    \includegraphics[width=\linewidth,trim={0.2cm 0.3cm 0.2cm 0.7cm},clip]{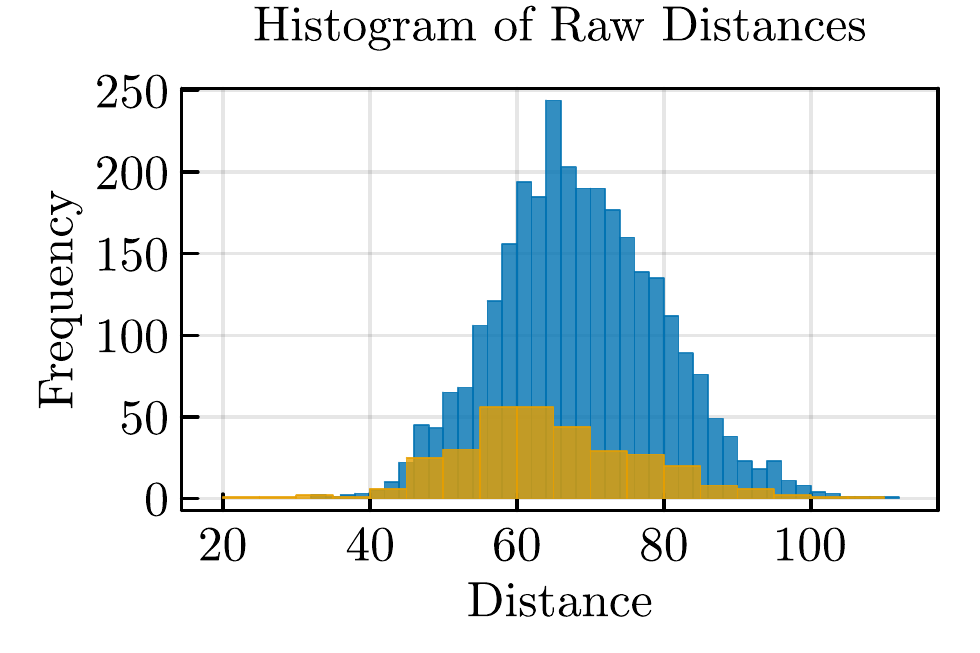}
    \caption{Pixelwise distances}
    \label{fig: histogram of coin distances pixelwise}
\end{subfigure}\hfill
\begin{subfigure}[t]{0.32\textwidth}
    \centering
    \includegraphics[width=\linewidth,trim={0.2cm 0.3cm 0.2cm 0.7cm},clip]{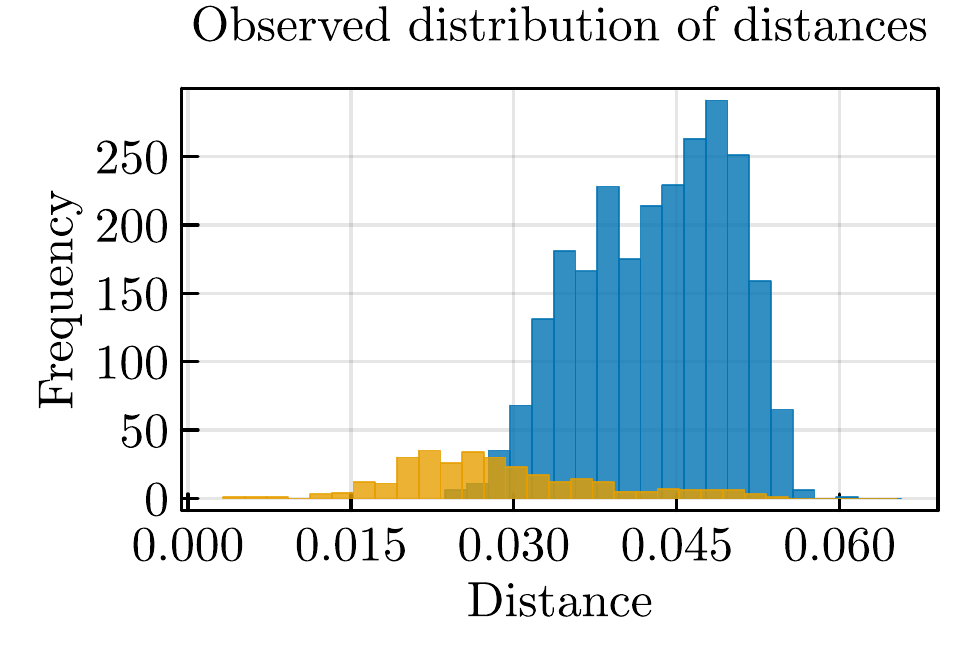}
    \caption{Distances computed using our method}
    \label{fig: histogram of coin distances sift}
\end{subfigure}\hfill
\begin{subfigure}[t]{0.32\textwidth}
    \centering
    \includegraphics[width=\linewidth,trim={0.2cm 0.3cm 0.2cm 0.7cm},clip]{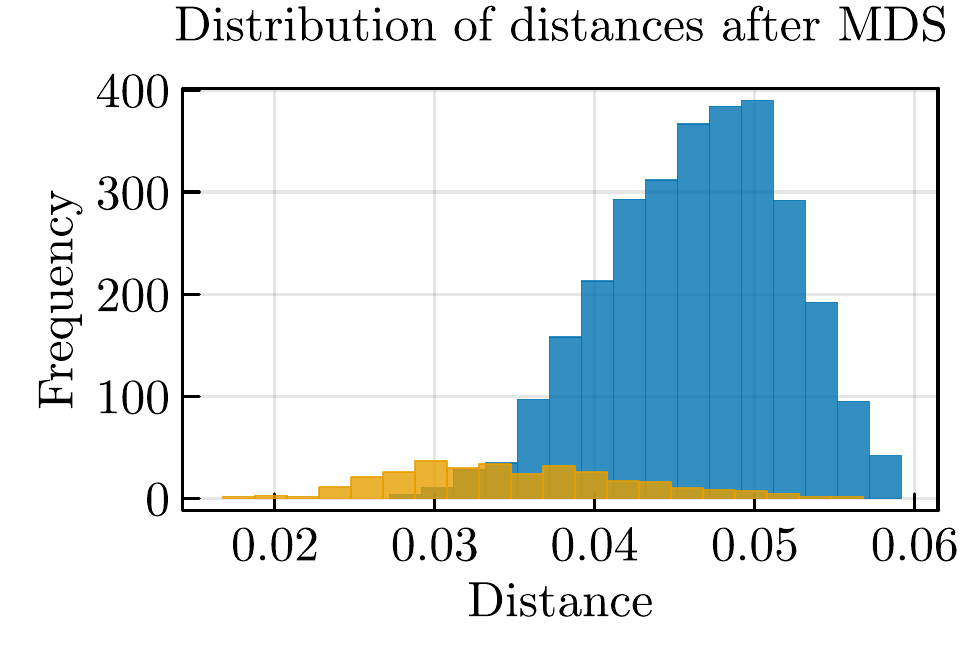}
    \caption{Distances after MDS embedding}
    \label{fig: histogram of coin distances MDS}
\end{subfigure}
\caption{Coin data: Histogram of within-cluster distances (orange) and inter-cluster distances (blue). The clusters correspond to the true clusters obtained by  a die study conducted by an expert numismatist.}
\label{fig: coin distances pixelwise vs sift}
\end{figure}

To identify comparable key features across pairs of coin images, we find sets of matched landmark pairs between images by exploiting the Scale-Invariant Feature Transform (SIFT) \citep{Lowe2004SIFT} and a low-distortion correspondence filtering procedure \citep{Lipman2014BoundedDistortion}. We use the landmark ranking algorithm of \citet{Gao2019Gaussian} and \citet{Gao2019GaussianII}, which we extend for ranking pairs of landmarks. We define a dissimilarity score between images using these ranked landmark pairs. \Cref{fig: coin landmarks} shows an example of matched landmark sets for images from the same die group and from different die groups. Details of the pipeline are provided in the Supplementary Material. 
\begin{figure}[h]
\begin{subfigure}[t]{0.48\textwidth}
    \centering
    \includegraphics[width=\linewidth]{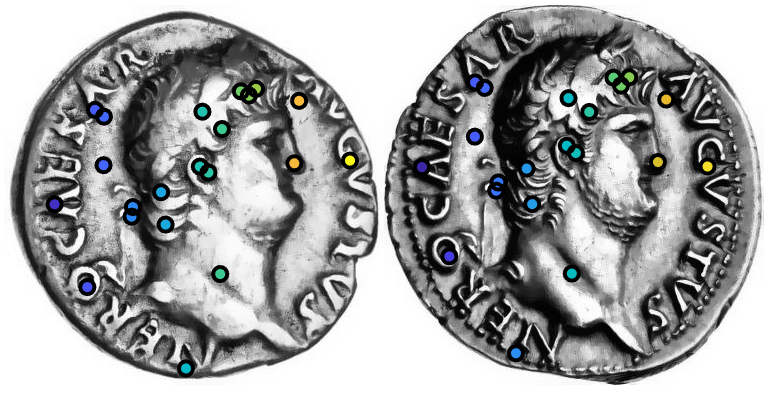}
    \caption{Matched landmarks on coins belonging to the same die group.}
    \label{fig: coin landmarks same die}
\end{subfigure}\hfill
\begin{subfigure}[t]{0.48\textwidth}
    \centering
    \includegraphics[width=\linewidth]{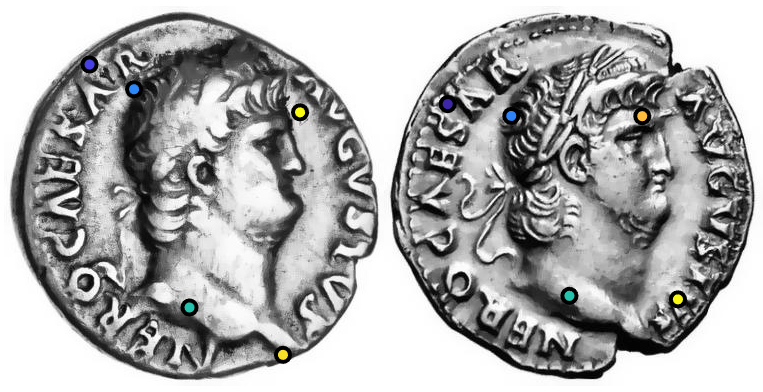}
    \caption{Matched landmarks on coins belonging to different die groups.}
    \label{fig: coin landmarks diff die}
\end{subfigure}
\caption{Matched sets of landmarks are used to construct a dissimilarity measure between images of coins. The number of landmarks is one of the components of this dissimilarity measure. Original unprocessed images are courtesy of the  \cite{ANSNero53} (left image in (a) and (b)), the \cite{CNGTriton} (right image in (a)), and \cite{GHN} (right image in (b)).}
\label{fig: coin landmarks}
\end{figure}
These dissimilarity scores are used as input for our algorithm. The pairwise dissimilarities are shown in \Cref{fig: histogram of coin distances sift}, and in Supplementary Material we compare the prior predictive distribution on dissimilarities as implied by by our data-driven prior specification process to the kernel density estimate of the dissimilarities.

\subsection{Results}
We run our model on the coin data for 50000 iterations, discarding the first 10000 iterations as burnin. We compare our model to the Mixture of Finite Mixtures (MFM) model proposed by \cite{Miller2018Mixture} and a Dirichlet Process Mixture (DPM), both as implemented in the \texttt{Julia} package \texttt{BayesianMixtures.jl} \citep{MFM}, using Normal mixture kernels with diagonal covariance matrix and conjugate priors. For the purposes of the comparison, we embed the coins as points in Euclidean space by applying Multi-Dimensional Scaling \citep{Kruskal1964MDS} (as implemented by the \texttt{MultivariateStats.jl} package in \texttt{Julia}) to the dissimilarity scores between the coins. We run the MFM and DPM samplers on the MDS output, which is 80-dimensional. 

To evaluate algorithm performance, we compute the \emph{co-clustering matrix} whose entries $s_{ij}$ are given by $s_{ij} = \Pr(z_i = z_j \mid \boldsymbol{X})$. Each $s_{ij}$ can be estimated from the MCMC output and its estimate is not affected by the label-switching phenomenon \citep{Stephens2000labelswitching,Fritsch2009coclustering}. In \Cref{fig: posterior coclustering coins} we compare the true adjacency matrix to the co-clustering matrices obtained by our model, MFM, and DPM. In \Cref{fig: histogram K coins} we show the marginal prior predictive distribution on the number of clusters $K$ implied by our choice of prior hyperparameters, as well as the posterior distribution on $K$ for each method. In Supplementary Material we show the posterior co-clustering matrix for our model without repulsion, the posterior distributions of $r$ and $p$, and we provide convergence diagnostics for the sampler.

For each method a clustering point-estimate is obtained via the SALSO algorithm \citep{Dahl2022salso}. This algorithm takes as its input the posterior samples of cluster allocations and searches for a point estimate that minimises the posterior expectation of the Variation of Information distance \citep{meila2007VI, Wade2018VIloss}. Point estimates are also obtained via $k$-means (on the MDS output) and $k$-medoids (on the dissimilarities), as implemented in the \texttt{Clustering.jl} package in \texttt{Julia}, using the value of $K$ obtained by the elbow method as in \Cref{sec: cluster-specific parameters}. \Cref{table: summary of point estimates coins} shows the Binder loss, Normalised Variation of Information (NVI) distance, Adjusted Rand Index (ARI), and Normalised Mutual Information (NMI) of these point estimates with respect to the true clustering. In Supplementary Material we show the adjacency matrices for the various point estimates.

The findings from this application suggest that clustering with our model on the distances alone can produce sensible results in terms of the original data, providing a viable strategy for high-dimensional settings. We further demonstrate the effectiveness of our model through simulation studies in Supplementary Material. These studies (1) show the effect of dimensionality and cluster separation on posterior inference, (2) demonstrate the robustness of our method to choice of prior hyperparameters, and (3) demonstrate the importance of the repulsion term in our likelihood. We remind the reader that our model underestimates uncertainty as discussed in \cref{sec: posterior uncertainty}.

\begin{figure}[h]
\centering
\begin{subfigure}[t]{0.5\textwidth}
    \centering
    \includegraphics[width=\linewidth,trim={0 0 0 0.9cm},clip]{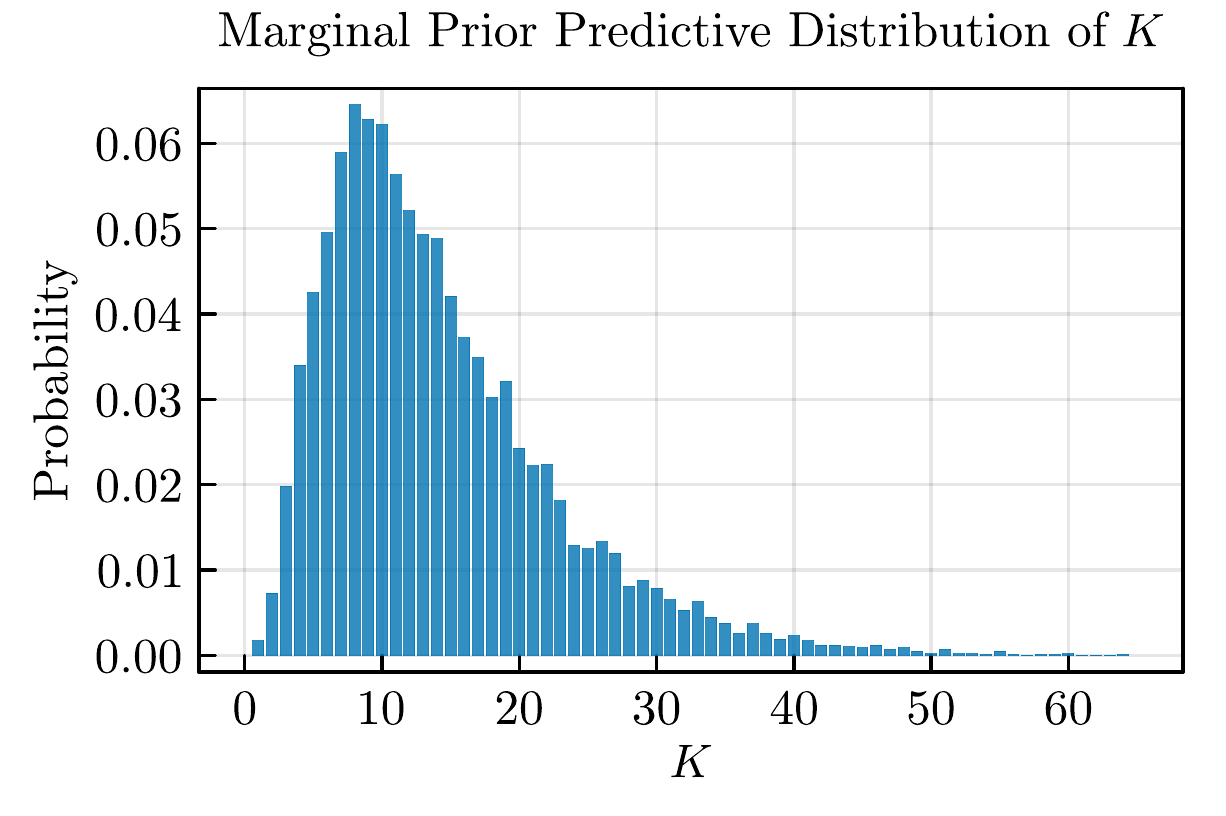}
    \caption{Marginal prior predictive distribution on $K$}
    \label{fig: prior predictive K coins}
\end{subfigure}\hfill
\begin{subfigure}[t]{0.5\textwidth}
    \centering
    \includegraphics[width=\linewidth,trim={0 0 0 0.9cm},clip]{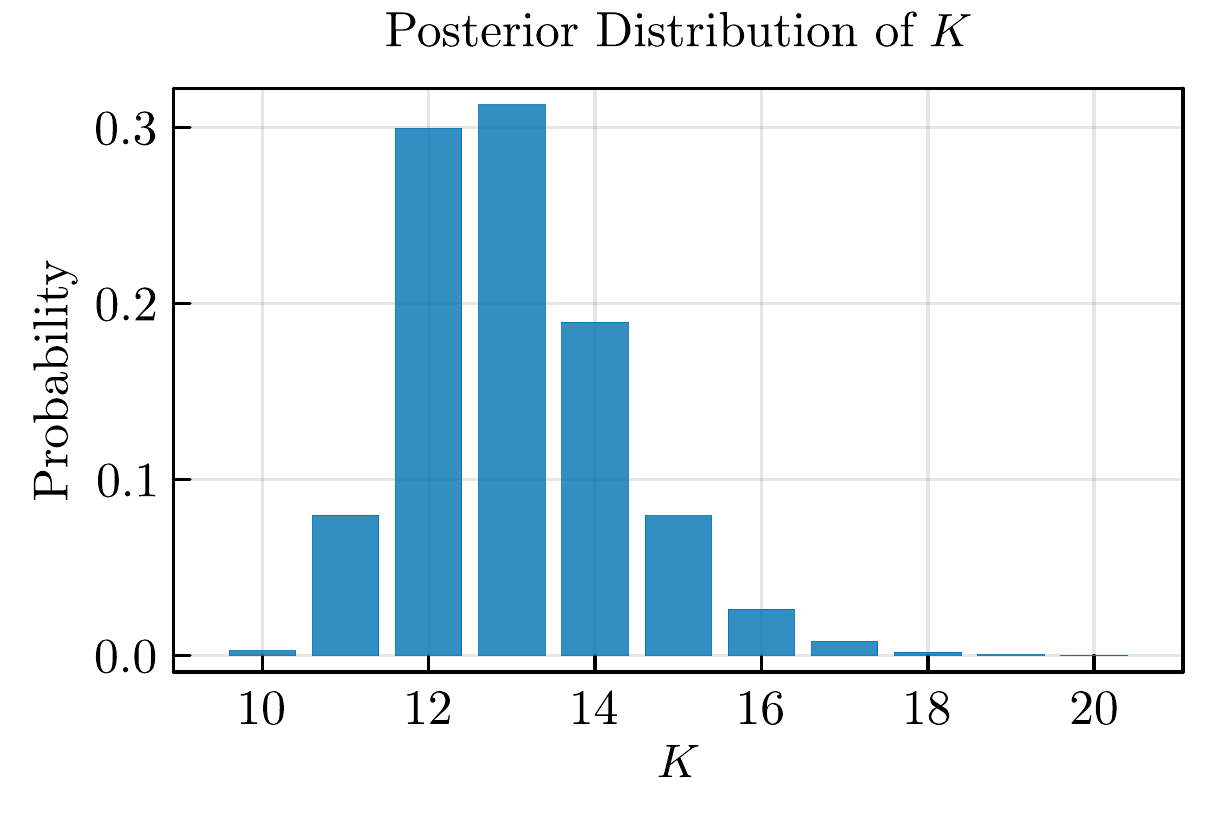}
    \caption{Our model}
    \label{fig: our model histogram K coins}
\end{subfigure}\vskip 12pt
\begin{subfigure}[t]{0.5\textwidth}
    \centering
    \includegraphics[width=\linewidth,trim={0 0 0 0.9cm},clip]{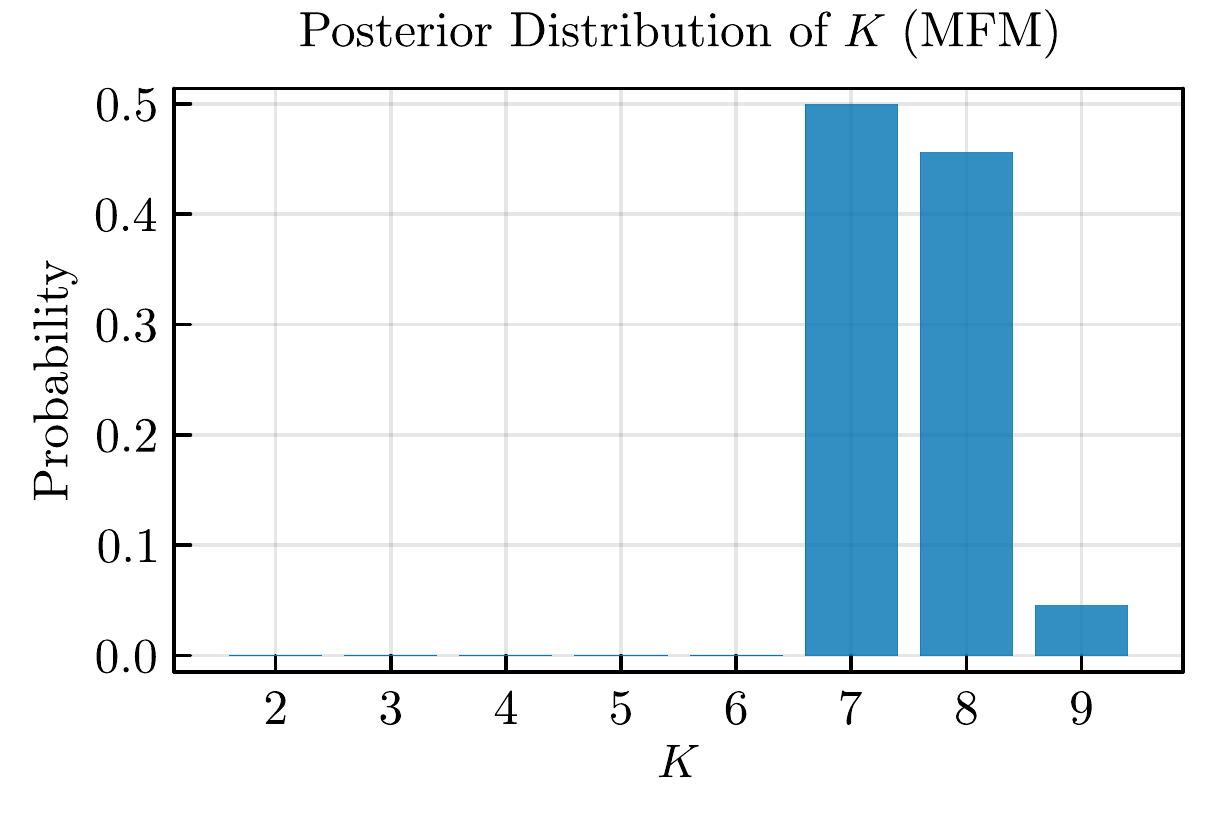}
    \caption{MFM}
    \label{fig: MFM histogram K coins}
\end{subfigure}\hfill
\begin{subfigure}[t]{0.5\textwidth}
    \centering
    \includegraphics[width=\linewidth,trim={0 0 0 0.9cm},clip]{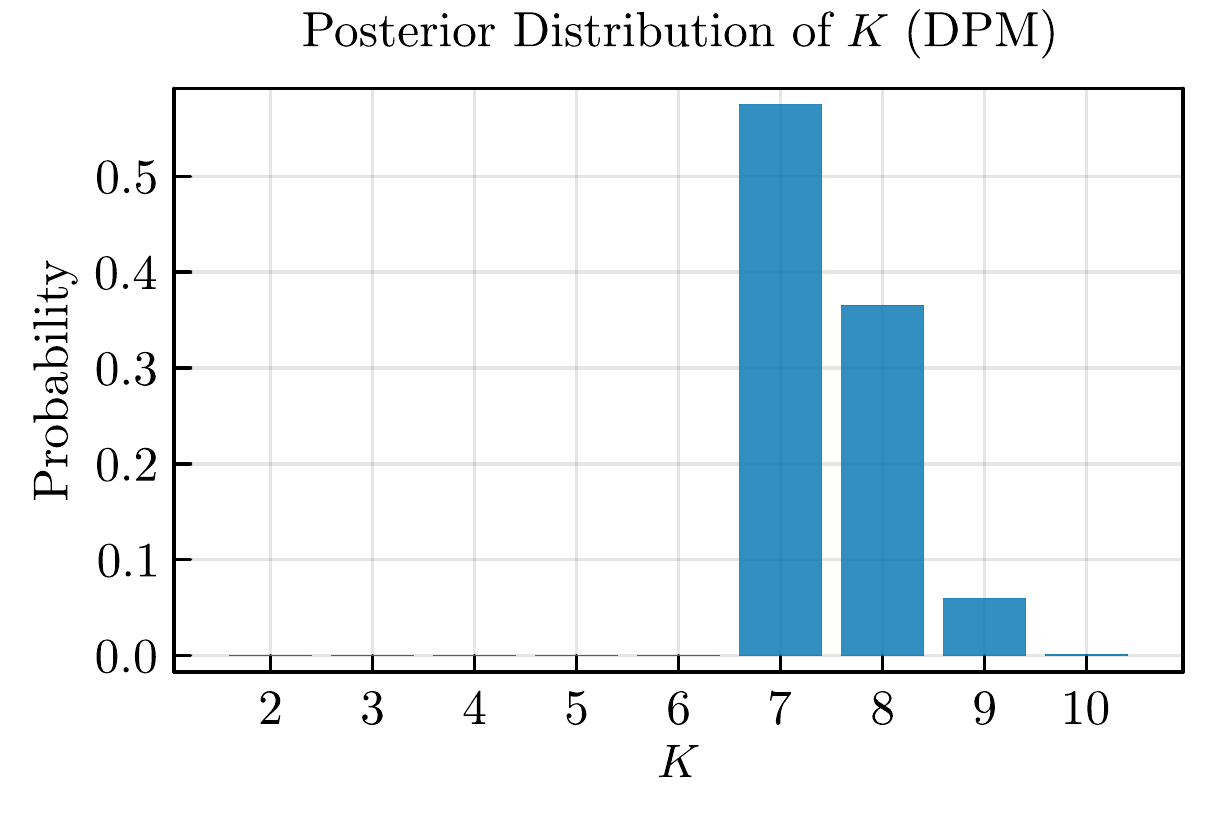}
    \caption{DPM}
    \label{fig: DPM histogram K coins}
\end{subfigure}
\caption{Coins data: Posterior distribution of the number of clusters $K$.}
\label{fig: histogram K coins}
\end{figure}
\begin{table}[h]
    \centering
    \spacingset{1.25}\begin{tabular}{| r | c c c c c|}
    \hline
    \multicolumn{1}{|c}{} & \multicolumn{1}{c}{Our Model} & MFM & DPM & $k$-means & $k$-medoids\\
    \hline
    Binder loss & \textbf{0.04} & 0.12 & 0.12 & 0.08 & 0.10\\
    NVI distance & \textbf{0.12} & 0.19 & 0.19 & 0.28 &  0.38\\ 
    ARI & \textbf{0.78} & 0.52 & 0.52 & 0.52 & 0.41\\
    NMI & \textbf{0.88} & 0.79 & 0.79 & 0.74 & 0.64\\
    K & 17 & 7 & 7 & 12 & 12\\
    \hline
\end{tabular}\spacingset{1.25}
\caption{Coins data: Comparison of point estimates with the true clustering. We highlight the best value for each measure in bold. }
\label{table: summary of point estimates coins}
\end{table}

\begin{figure}[h]
\centering
\begin{subfigure}[t]{0.5\textwidth}
    \centering
    \includegraphics[width=\linewidth,trim={0.8cm 0 0.8cm 0.9cm},clip]{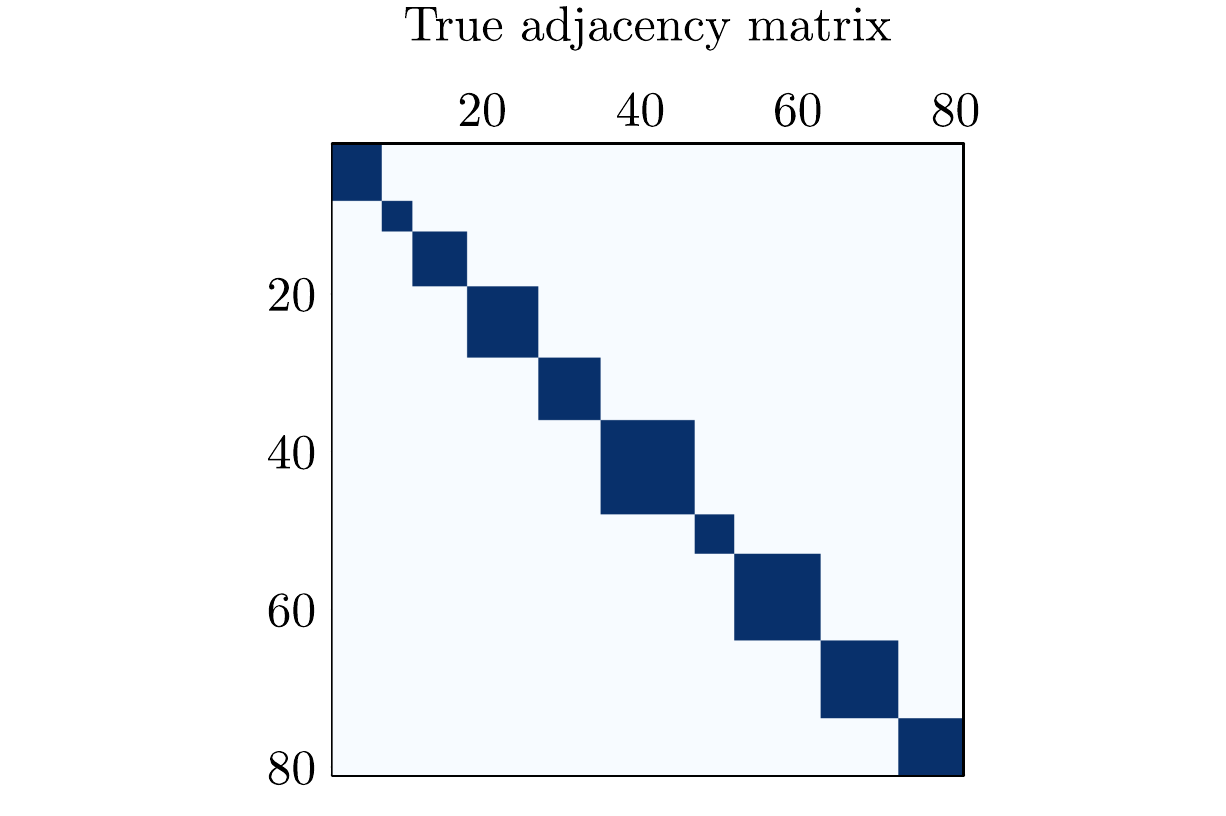}
    \caption{Adjacency matrix of the true clustering}
    \label{fig: adjacency matrix coins}
\end{subfigure}\hfill
\begin{subfigure}[t]{0.5\textwidth}
    \centering
    \includegraphics[width=\linewidth,trim={0.8cm 0 0.8cm 0.9cm},clip]{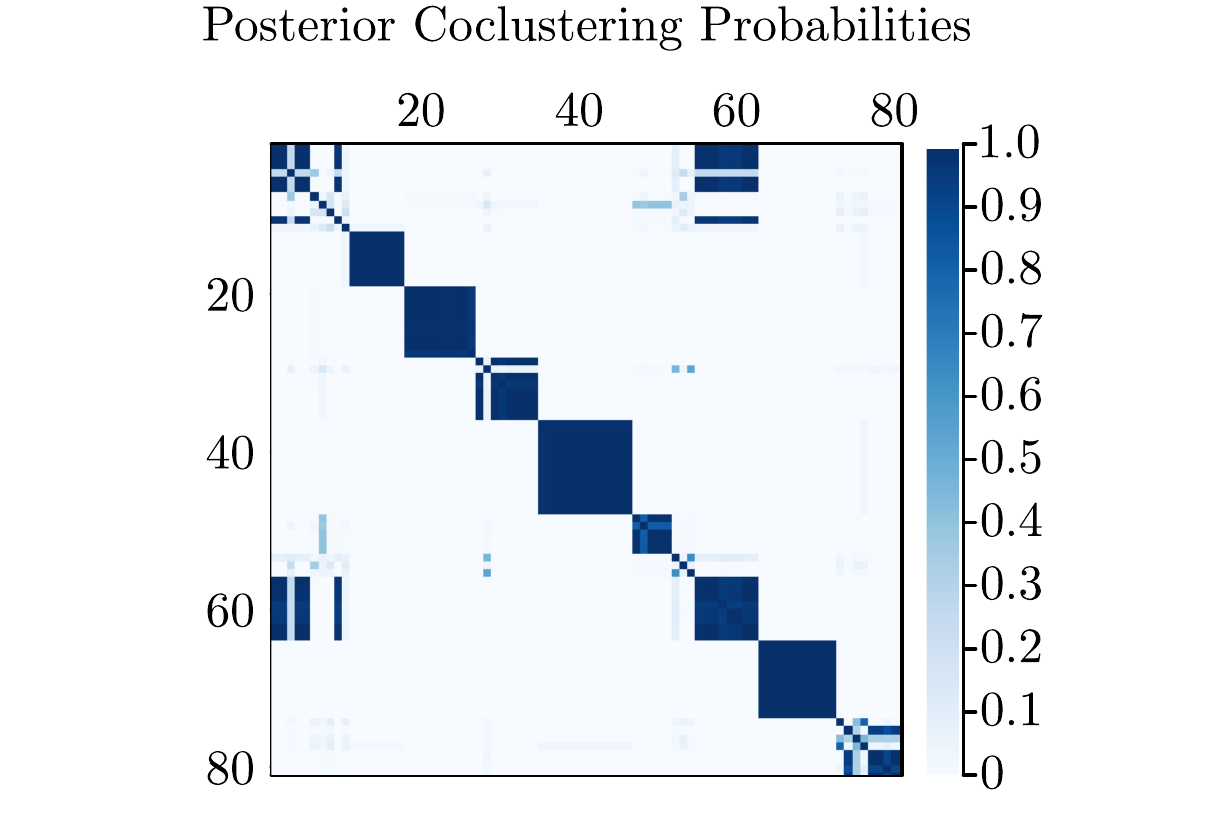}
    \caption{Our model}
    \label{fig: our posterior coclustering matrix coins}
\end{subfigure}\vskip 12pt
\begin{subfigure}[t]{0.5\textwidth}
    \centering
    \includegraphics[width=\linewidth,trim={0.8cm 0 0.8cm 0.9cm},clip]{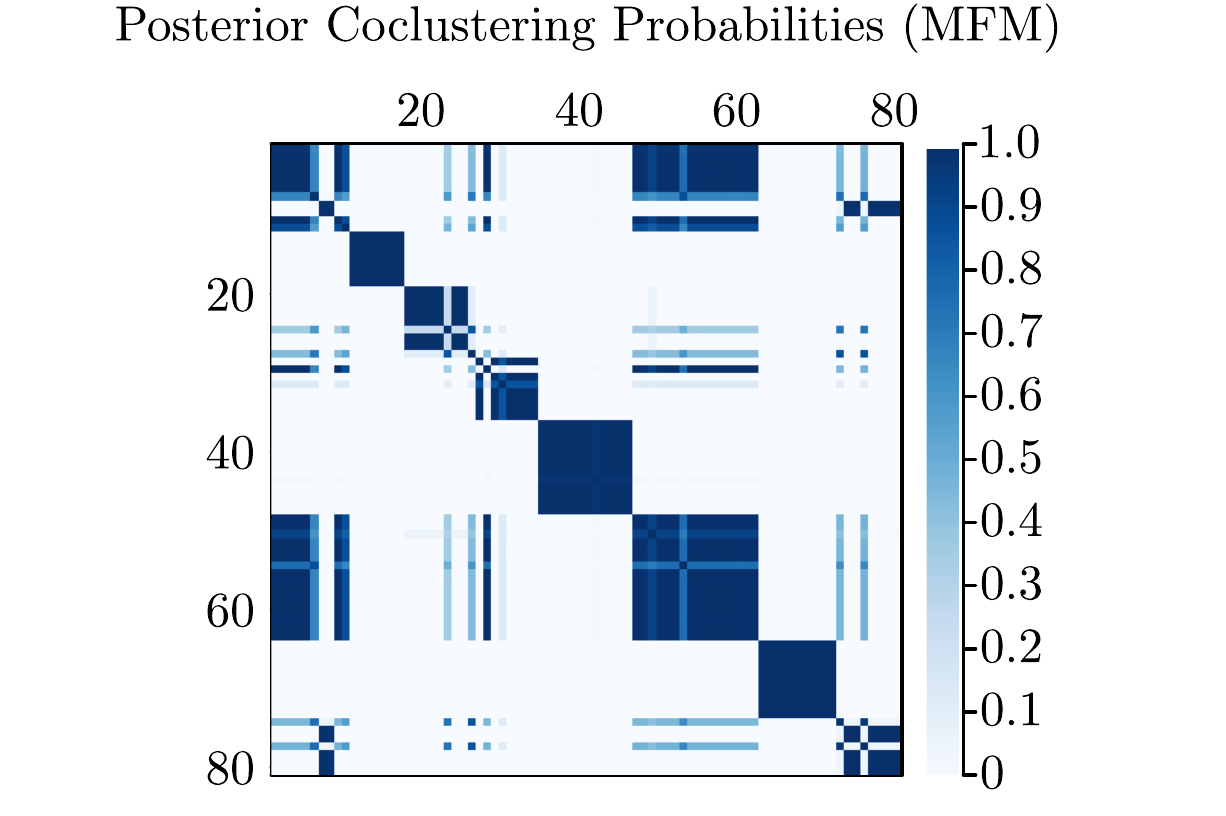}
    \caption{MFM}
    \label{fig: MFM coclustering matrix coins}
\end{subfigure}\hfill
\begin{subfigure}[t]{0.5\textwidth}
    \centering
    \includegraphics[width=\linewidth,trim={0.8cm 0 0.8cm 0.9cm},clip]{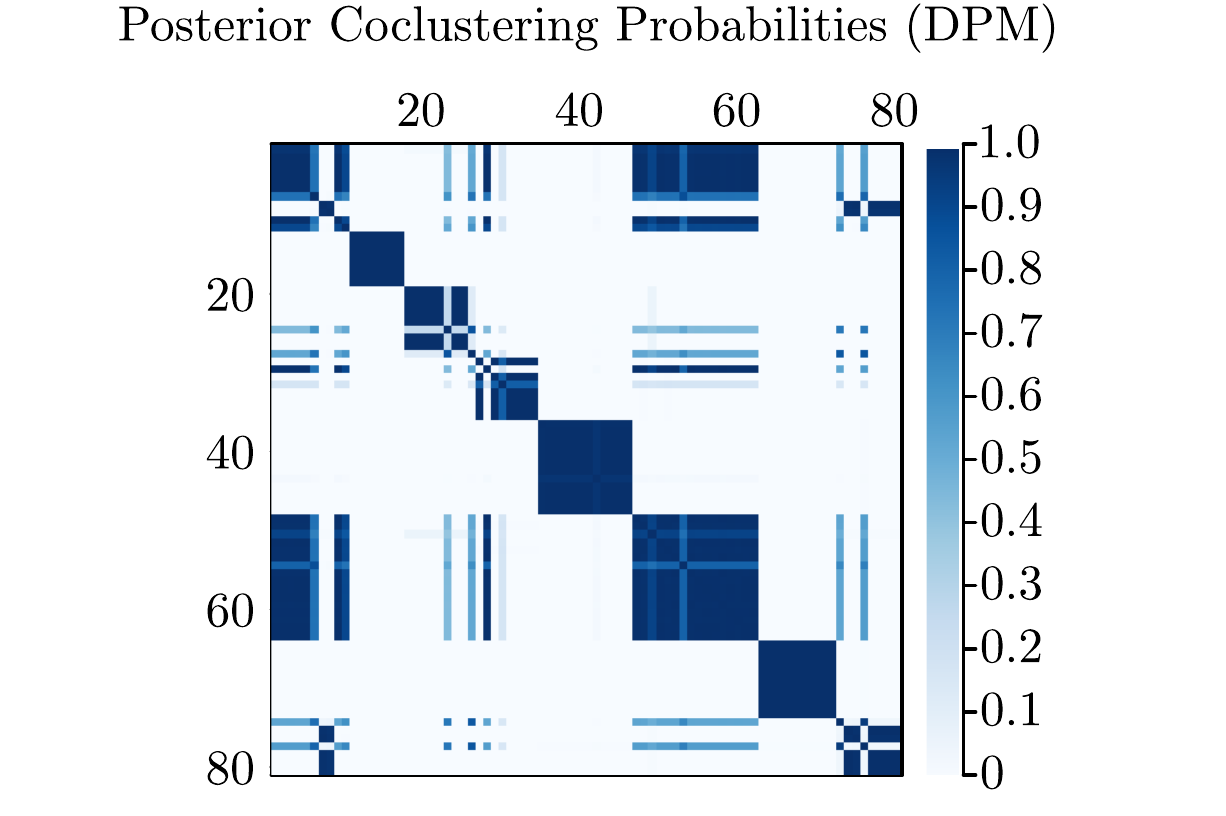}
    \caption{DPM}
    \label{fig: DPM coclustering matrix coins}
\end{subfigure}
\caption{Coins data: Posterior co-clustering matrices}
\label{fig: posterior coclustering coins}
\end{figure}
\clearpage

\section{Discussion}\label{sec: discussion}
\emph{It is the curse of dimensionality, a malediction that has plagued the scientist from the earliest days.}
\begin{flushright}
--- Richard Bellman, 1961
\end{flushright}
When clustering in large dimensions, the main statistical challenges include (i) accommodating for the sparsity of points in large dimensions, (ii) capturing the complexity of the data generating process in high-dimensions, including the interdependence of features which might be cluster-specific, (iii) estimating model parameters, (iv) developing methods which are robust to different underlying generative processes and cluster characteristics, (v) devising computational algorithms that are scalable to large datasets, (vi) producing interpretable results, (vii) assessing the validity of the cluster allocation, and (viii) determining the degree to which different features contribute to clustering. 

We propose a hybrid clustering method for high-dimensional problems which is essentially model-based clustering on pairwise distances between the original observations. This method is also applicable in settings where the likelihood is not computationally tractable. The strategy allows us to overcome many of the aforementioned challenges, bypassing the specification of a model on the original data. The main contribution of our work is to combine cohesive and repulsive components in the likelihood, and we provide theoretical justifications for our model choices. Our method is robust to different generative processes, and computationally more efficient than model-based approaches for high dimensions because it reduces the multi-dimensional likelihood on each data point to a unidimensional likelihood on each distance. Our method also leads to interpretable results as clusters are defined in terms of the original observations. The model can be easily extended to categorical variables by considering (for example) the Hamming distance or cross entropy and specifying appropriate distributions $f$ and $g$ in \Cref{eqn: likelihood general form}. The main drawbacks of our methodology is that the role of each feature is embedded in the distances and model performance is dependent on the definition of the distances. We do not advise the use of a distance-based approach when the dimension is small because in that case standard model-based approaches work well, and using distances as a summary of the data causes loss of information. 

From our application in digital numismatics, it is clear that the definition of the distances plays a crucial role and a future direction of this work is to develop landmark estimation methods better able to capture the distinguishing features of images.

Finally, there is an interesting connection between the likelihood in \Cref{eqn: likelihood general form} and the likelihood for a stochastic blockmodel in the $p1$ family \citep{Snijders1997Graphs,Nowicki2001Blockstructures,Schmidt2013Networks} where every block of nodes can be thought of as a cluster of similar objects. This connection is a topic of further research.

\section{Supplementary Material}
\begin{description}
\item[Additional results:] Proof of \Cref{thm: distributions on K}, detailed description of the MCMC algorithm, computational complexity of the MCMC algorithm, details of the method to compute distances between coin images, additional plots and results from the numismatic example, alternative method to choose prior hyperparameters, and simulation studies are available online in a supplementary PDF document.
\item[Code:] Julia code to perform posterior inference is available on Github at \url{https://github.com/abhinavnatarajan/RedClust.jl}, and also as a package in the default Julia package registry (``General"). The code used to run the numismatic and simulated examples and generate the corresponding figures is available at \url{https://github.com/abhinavnatarajan/RedClust.jl/tree/examples}.
\end{description}

\section{Funding}
This work was supported by the Singapore Ministry of Education through the Academic Research Fund Tier 2 under grant MOE-T2EP40121-0021, the National University of Singapore through the NUS HSS Seed Fund under grant R-607-000-449-646, and by Yale-NUS College through Yale-NUS IG20-RA001.

\section{Conflicts of Interest}
We report that there are no competing interests to declare. 

\medskip
\bibliographystyle{jss2} 
\bibliography{references}
\end{document}